\documentclass[aps,prl,twocolumn,superscriptaddress]{revtex4-1}
\usepackage{graphicx}
\usepackage{subfigure}
\usepackage{epstopdf}
\usepackage{amsmath}
\usepackage{amssymb}
\usepackage{amsfonts}
\usepackage{mathrsfs}
\usepackage{bbm}
\usepackage{url}
\usepackage[T1]{fontenc}
\usepackage{csquotes}
\MakeOuterQuote{"}

\usepackage{algorithm}
\usepackage{algorithmicx}
\usepackage{algpseudocode}

\usepackage{dcolumn}
\usepackage{color}
\usepackage{graphicx}
\usepackage{subfigure}
\usepackage{epstopdf}

\usepackage{amsmath}
\usepackage{amssymb}
\usepackage{amsfonts}
\usepackage{amsthm}
\usepackage{mathrsfs}
\usepackage{bbm}

\newtheorem{lemma}{Lemma}

\usepackage{url}
\usepackage[T1]{fontenc}
\usepackage{csquotes}

\usepackage{algorithm}
\usepackage{algorithmicx}
\usepackage{algpseudocode}

\usepackage{dcolumn}
\usepackage{color}

\def\vec#1{\boldsymbol{#1}} 

\newcommand{\tr}{\operatorname{tr}}

\newcommand{\caA}{\mathcal{A}}
\newcommand{\caP}{\mathcal{P}}

\newcommand{\rmH}{\mathrm{H}}
\newcommand{\rmT}{\mathrm{T}}
\newcommand{\rmV}{\mathrm{V}}

\newcommand{\Sym}{\mathrm{Sym}}

\newcommand{\scrA}{\mathscr{A}}
\newcommand{\scrB}{\mathscr{B}}

\newcommand{\scrE}{\mathscr{E}}

\newcommand{\scrK}{\mathscr{K}}

\newcommand{\caH}{\mathcal{H}}

\newcommand{\bbone}{\mathbbm{1}}
\newcommand{\bbW}{\mathbb{W}}

\newcommand{\be}{\begin{equation}}
\newcommand{\ee}{\end{equation}}
\newcommand{\ba}{\begin{align}}
\newcommand{\ea}{\end{align}}

\def\<{\langle}  
\def\>{\rangle}  

\def\eqref#1{\textup{(\ref{#1})}}  
\newcommand{\eref}[1]{Eq.~\textup{(\ref{#1})}}

\newcommand{\fref}[1]{Fig.~\ref{#1}}

\newcommand{\lref}[1]{Lemma~\ref{#1}}

\newcommand{\cref}[1]{Conjecture~\ref{#1}}
\newcommand{\Cref}[1]{Conjecture~\ref{#1}}

\definecolor{ngreen}{rgb}{0.2,0.6,0.2}
\definecolor{ngold}{rgb}{0.7,0.6,0.2}

\def\vec#1{\boldsymbol{#1}} 

\def\<{\langle}  
\def\>{\rangle}  

\def\eqref#1{\textup{(\ref{#1})}}  

\newcommand{\rcite}[1]{Ref.~\cite{#1}}
\newcommand{\rscite}[1]{Refs.~\cite{#1}}

\begin{document}
\title{Experimental Realization of Genuine Three-copy Collective Measurements for Optimal Information Extraction}

	\author{Kai Zhou} 
	\affiliation{Laboratory of Quantum Information, University of Science and Technology of China, Hefei 230026, People's Republic of China}
	\affiliation{CAS Center For Excellence in Quantum Information and Quantum Physics, University of Science and Technology of China, Hefei 230026, People's Republic of China}

  \author{Changhao Yi}
    \affiliation{State Key Laboratory of Surface Physics, Department of Physics, and Center for Field Theory and Particle Physics, Fudan University, Shanghai 200433, China}
    \affiliation{Institute for Nanoelectronic Devices and Quantum Computing, Fudan University, Shanghai 200433, China}
    \affiliation{Shanghai Research Center for Quantum Sciences, Shanghai 201315, China}

    	\author{Wen-Zhe Yan} 
	\affiliation{Laboratory of Quantum Information, University of Science and Technology of China, Hefei 230026, People's Republic of China}
	\affiliation{CAS Center For Excellence in Quantum Information and Quantum Physics, University of Science and Technology of China, Hefei 230026, People's Republic of China}
 
	\author{Zhibo Hou} 
	\email{houzhibo@ustc.edu.cn}
	\affiliation{Laboratory of Quantum Information, University of Science and Technology of China, Hefei 230026, People's Republic of China}
	\affiliation{CAS Center For Excellence in Quantum Information and Quantum Physics, University of Science and Technology of China, Hefei 230026, People's Republic of China}
 \affiliation{Hefei National Laboratory, Hefei 230088, People's Republic of China}
 
    \author{Huangjun~Zhu}
    \email{zhuhuangjun@fudan.edu.cn}
    \affiliation{State Key Laboratory of Surface Physics, Department of Physics, and Center for Field Theory and Particle Physics, Fudan University, Shanghai 200433, China}
    \affiliation{Institute for Nanoelectronic Devices and Quantum Computing, Fudan University, Shanghai 200433, China}
    \affiliation{Shanghai Research Center for Quantum Sciences, Shanghai 201315, China}
     \affiliation{Hefei National Laboratory, Hefei 230088, People's Republic of China}
    
	\author{Guo-Yong Xiang}
	\email{gyxiang@ustc.edu.cn}
	\affiliation{Laboratory of Quantum Information, University of Science and Technology of China, Hefei 230026, People's Republic of China}
	\affiliation{CAS Center For Excellence in Quantum Information and Quantum Physics, University of Science and Technology of China, Hefei 230026, People's Republic of China}
 \affiliation{Hefei National Laboratory, Hefei 230088, People's Republic of China}
 
	\author{Chuan-Feng Li}
	\affiliation{Laboratory of Quantum Information, University of Science and Technology of China, Hefei 230026, People's Republic of China}
	\affiliation{CAS Center For Excellence in Quantum Information and Quantum Physics, University of Science and Technology of China, Hefei 230026, People's Republic of China}
 \affiliation{Hefei National Laboratory, Hefei 230088, People's Republic of China}
 
	\author{Guang-Can Guo}
	\affiliation{Laboratory of Quantum Information, University of Science and Technology of China, Hefei 230026, People's Republic of China}
	\affiliation{CAS Center For Excellence in Quantum Information and Quantum Physics, University of Science and Technology of China, Hefei 230026, People's Republic of China}
 \affiliation{Hefei National Laboratory, Hefei 230088, People's Republic of China}

\begin{abstract}
Nonclassical phenomena tied to entangled states are the focus of foundational studies and powerful resources in many applications. By contrast, the counterparts in quantum measurements are still poorly understood. Notably, genuine multipartite nonclassicality is barely discussed, let alone its experimental realization.
Here we experimentally demonstrate the power of genuine tripartite nonclassicality in quantum measurements based on a simple estimation problem. To this end we realize an optimal genuine three-copy collective measurement via a nine-step two-dimensional photonic quantum walk with 30 elaborately designed coin operators.
Then we realize an optimal estimation protocol and achieve an unprecedented high estimation fidelity, which can beat all strategies based on restricted collective measurements by more than 11 standard deviations.  These results clearly demonstrate that genuine collective measurements can extract more information than local measurements and restricted collective measurements. Our work opens the door for exploring genuine multipartite nonclassical measurements and their power in quantum information processing.
\end{abstract}

\date{\today}
 
\maketitle

\emph{Introduction}---The cross fertilization of quantum mechanics and information theory has revolutionized the way information is perceived and processed. This revolution is closely tied to a number of intriguing nonclassical phenomena predicted by quantum theory. For example, local measurements on entangled states can produce correlations that are much stronger than what can be produced by classical means \cite{Bell64,BrunCPS14}. Such nonclassical correlations are interesting not only to foundational studies, but also to many practical tasks, such as quantum cryptography and distributed information processing.

Although nonclassicality in quantum states has been actively studied by numerous researchers, nonclassicality in quantum measurements, often described by positive operator-valued measures (POVMs) \cite{NielC00book}, is still poorly understood. A few notable exceptions are tied to collective measurements on 
 two or more independent copies of a quantum state.  Surprisingly, it turns out that collective measurements can extract much more information than local measurements, although there is no entanglement or correlations among the different copies \cite{PereW91,MassP95,BagaBGM06S,Zhu12the}. This superior information extraction capability is rooted in the nonclassicality of quantum measurements rather than quantum states. It may offer practical advantages in many tasks, such as quantum state estimation \cite{MassP95,BagaBGM06S,Zhu12the,ZhuH18U},  direction estimation \cite{GisiP99},   multiparameter estimation \cite{VidrDGJ14,Lu21incorporating,chen22information}, and quantum state discrimination \cite{PereW91,higgins11multiple,Martines21quantum,conlon2023discriminating,Tian24minimum}. Recently, the nonclassicality of collective measurements has been experimentally demonstrated on photonic \cite{hou2018deterministic,experimental20tang,Parniak18beating,Wu20minimizing,Tian24minimum}, ion-trap \cite{Conlon23approaching}, and superconducting platforms \cite{Conlon23approaching,conlon2023discriminating}. Nevertheless, all these demonstrations are limited to the simplest two-copy collective measurements because of the tremendous difficulty in realizing multicopy collective measurements as pointed out in \rscite{Conlon23approaching,conlon2023discriminating}.

The situation is more complicated and interesting if we turn to the multipartite scenario, which allows a variety of nonclassical phenomena. Genuine multipartite nonclassicality is particularly appealing as it cannot be reduced to nonclassicality of fewer parties. In the case of quantum states, it is well known in the form of genuine multipartite entanglement (GME) \cite{greenberger1989going,greenberger1990bell,GUHNE2009entanglement,xie21triangle}, which cannot be reduced to entanglement among fewer parties. GME plays a central role in 
nonstatistical tests of quantum nonlocality \cite{greenberger1989going,greenberger1990bell,pan2000experimental},   quantum networking \cite{contreras21genuine}, quantum cryptography \cite{Das21universal}, and quantum metrology \cite{hyllus12fisher,Toth12multipartite}. By contrast, much less is known about genuine multipartite nonclassicality in quantum measurements and its significance in information processing. Actually, we are not even aware of a clear definition in the literature.

\begin{figure*}[htbp]
\centering\includegraphics[width=0.9\linewidth]{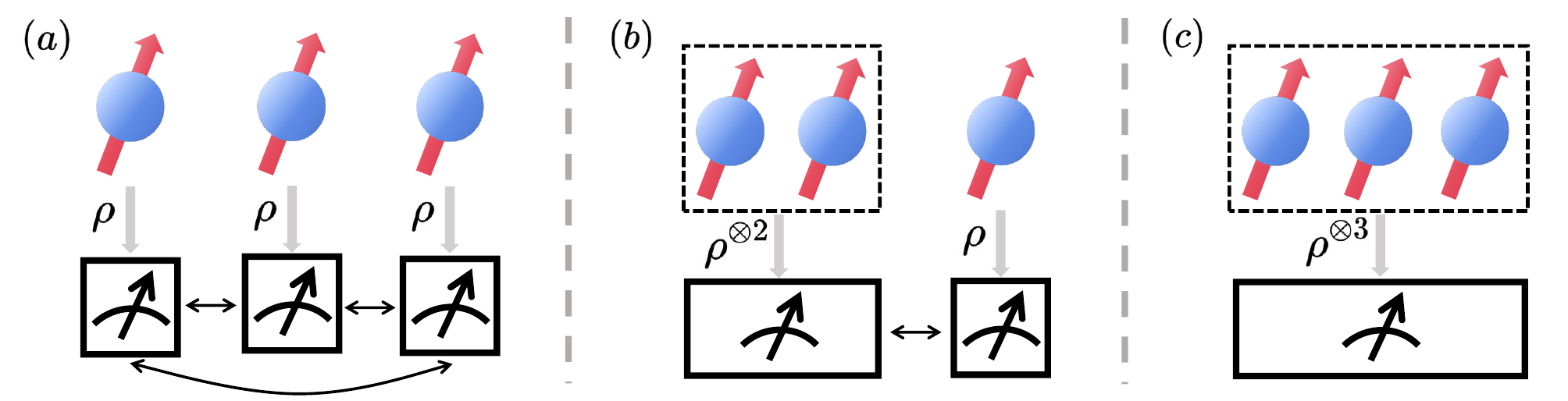}
	\caption{\label{fig:figure1}  Three types of quantum measurements on three copies of a quantum state $\rho$.  (a) Local measurements assisted by two-way classical communication (represented by double arrows). (b) 
Restricted collective measurements assisted by two-way classical communication (all such measurements are biseparable).  (c) Genuine collective measurements. 
 } 
\end{figure*}

In this work, we explore genuine tripartite nonclassicality in quantum measurements and experimentally demonstrate its information-extraction power  based on a simple estimation problem. More specifically, we consider the estimation of a random qubit pure state given three identically prepared copies. We first choose a genuine three-copy collective measurement that is optimal for the estimation task and is strictly better than alternative choices based on restricted collective measurements assisted by adaptive strategies.  Surprisingly, however, GME is not involved in any POVM element associated with this optimal measurement, which reveals a subtle distinction between quantum measurements and quantum states with regard to multipartite nonclassicality. 
Then we experimentally realize this genuine collective measurement with  high quality via a nine-step two-dimensional (2D)  photonic quantum walk featuring 30 elaborately designed coin operators. Based on this measurement we realize an optimal estimation protocol and achieve an unprecedented high estimation fidelity. 
Notably, our achieved estimation fidelity exceeds the upper bound for restricted collective measurements by over 11 standard deviations,  which demonstrates the true realization of a genuine three-copy collective measurement and its power in quantum state estimation.

\emph{Optimal qubit state estimation with a genuine three-copy collective measurement}---Suppose we are given three copies of an unknown qubit pure state as described by the density operator $\rho$ and are asked to estimate the identity of $\rho$ based on a suitable quantum measurement. The performance of an estimation strategy is characterized by the average estimation fidelity over potential measurement outcomes and the unknown state. If we can only perform local measurements on individual copies as illustrated in \fref{fig:figure1}(a), then the maximum estimation fidelity is $F=(3+\sqrt 3)/6\approx 0.7887$ \cite{PhysRevA.71.062318}. By contrast, if we can perform collective measurements on all three copies simultaneously as illustrated in \fref{fig:figure1}(c), then we can attain the maximum estimation fidelity $4/5$~\cite{MassP95}. Here we choose a special optimal measurement described by a POVM $\scrE=\{E_j\}_{j=1}^7$ with seven POVM elements \cite{PhysRevLett.81.1351}:
\begin{equation}\label{eq:measurement}
\begin{aligned}                   \!\!\!E_j&=\frac23|\psi_j\rangle\langle\psi_j|^{\otimes 3},\;\; j=1,\dots,6,\;\;
E_7&=\mathbb{I}-\sum_{j=1}^6E_j,
\end{aligned}
\end{equation}
where
\begin{equation}\label{eq:elements}
\begin{aligned}                      
|\psi_1\rangle=|0\rangle,&\quad |\psi_2\rangle=|1\rangle,\\
|\psi_3\rangle=\frac1{\sqrt2}(|0\rangle+|1\rangle),&\quad |\psi_4\rangle=\frac1{\sqrt2}(|0\rangle-|1\rangle),\\
|\psi_5\rangle=\frac1{\sqrt2}(|0\rangle+i|1\rangle),&\quad |\psi_6\rangle=\frac1{\sqrt2}(|0\rangle-i|1\rangle)
\end{aligned}
\end{equation}
are six single-qubit states which form a regular octahedron when represented on the Bloch sphere. 

To appreciate the significance of the optimal measurement characterized by the POVM $\scrE$, we need to introduce some additional concepts. Let $\caP$ be a given bipartition, say $(12|3)$, of the three parties associated with the three copies of $\rho$. A POVM (and the corresponding measurement)  is $\caP$ separable if every nonzero POVM element is $\caP$ separable (a $\caP$ separable state after normalization). A POVM  is \emph{biseparable} if it is a coarse-graining \cite{MartM90,Zhu22} of a convex combination of three POVMs that are $(12|3)$-separable, $(13|2)$-separable, and $(23|1)$-separable, respectively; see  Supplemental Material S4 \cite{supplementary} and companion paper \cite{YiZHX24} for more details.
Intuitively, all measurements that can be realized by restricted collective measurements assisted by classical communication, as illustrated in \fref{fig:figure1}(b), are biseparable. A measurement is  \emph{genuinely collective} if it is not biseparable.

As shown in  Supplemental Material S5, the POVM $\scrE$ defined in \eref{eq:measurement} is genuinely collective. Moreover, the maximum estimation fidelity of biseparable measurements is only
$(8 + \sqrt{22})/16\approx 0.7932$ according to the companion paper \cite{YiZHX24}. So genuine collective measurements are crucial to attaining the maximum estimation fidelity $4/5$. This result is quite surprising given that POVM elements $E_j$ for $j=1,2,\ldots,6$ correspond to product states, while $E_7$ is biseparable. Note that GME in POVM elements is not necessary for constructing a genuine collective POVM. This peculiar property is of intrinsic interest that is beyond the main focus of this work.

\begin{figure*}[htbp]	\centering\includegraphics[width=1\linewidth]{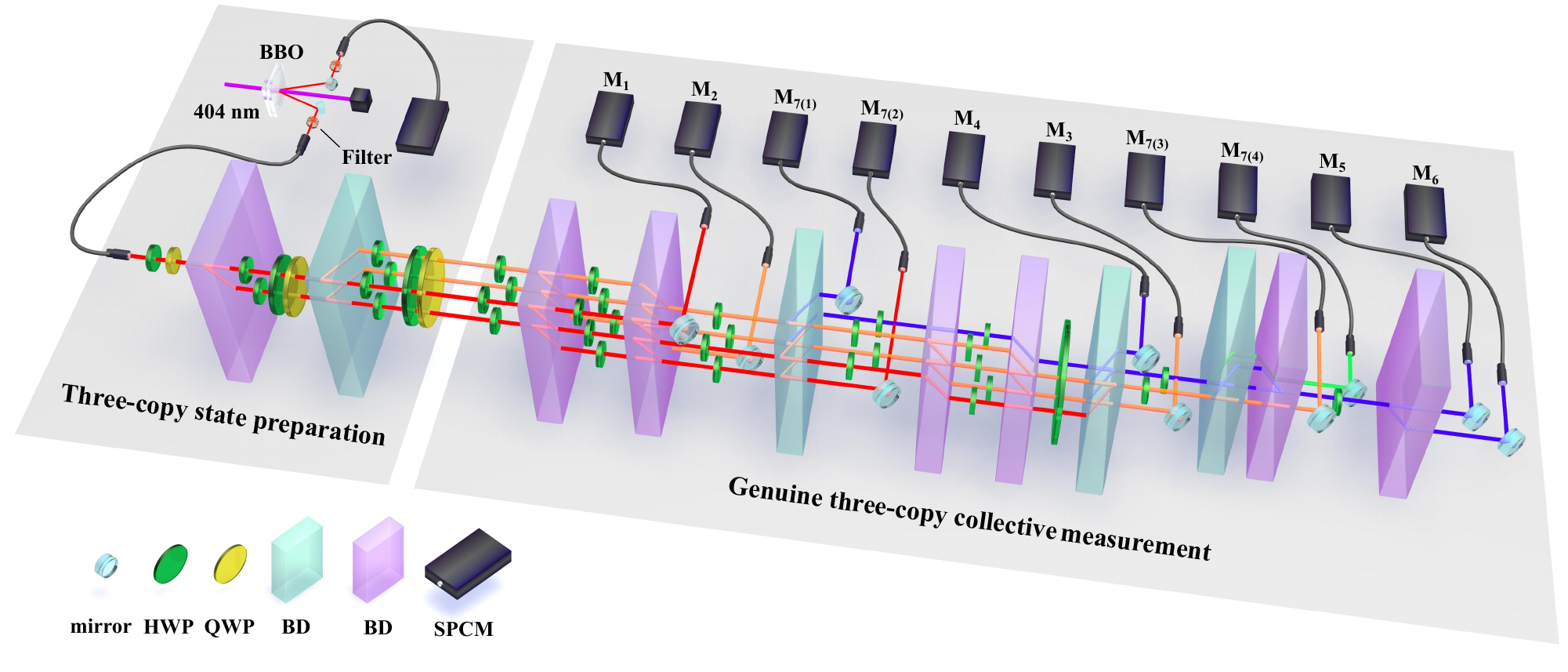}
	\caption{\label{fig:exp}  Experimental setup for realizing the genuine three-copy collective measurement in \eref{eq:measurement} and its application in quantum state estimation. The setup is composed of two modules. The preparation module encodes three copies of a qubit state in three degrees of freedom of a single photon. The measurement module realizes the genuine three-copy collective measurement via a nine-step 2D photonic quantum walk. Half-wave plates (HWPs) and quarter-wave plates (QWPs) are used to realize  30 engineered coin operators; two types of beam displacers (BDs) are used to realize two types of conditional transition operators, which correspond to horizontal walking (blue) and vertical walking (purple), respectively.  
		Beams with different colors between BDs mean different vertical positions.  The six single-photon-counting modules (SPCMs) $\mathrm{M_1}$ to $\mathrm{M_{6}}$ correspond to the first six outcomes in \eref{eq:measurement}, while
		the four SPCMs $\mathrm{M_{7(1)}}$ to $\mathrm{M_{7(4)}}$ together correspond to the last outcome in \eref{eq:measurement}.}
\end{figure*}

\emph{Implementing the genuine three-copy collective measurement via a 2D photonic quantum walk}---Quantum walks are a powerful tool for realizing general quantum measurements \cite{KurzW13,li19implementation,Wang23generalized}, which have been experimentally realized on single-qubit \cite{BianLQZ15,ZhaoYKX15} and two-qubit \cite{hou2018deterministic,experimental20tang,Wu20minimizing,yuan2020direct} photonic systems. Here, by extending from 1D- to 2D-quantum walks, we show how to realize quantum measurements on a three-qubit system. In a 2D quantum walk, the system state is effectively described by three degrees of freedom, i.e., the vertical and horizontal positions of the walker, denoted by $y,x=\dots,-1,0,1,\dots$, respectively, and the coin state represented by $c = 0, 1$. At step $t$, the system evolution is determined by a unitary operator of the form $U(t) = T(t)C(t)$. Here  $T(t)$  is a conditional translation operator and may take on one of the following two forms:
\begin{equation}
\begin{aligned}
\!\!\! T_\rmV & =\sum_y\left(|y+1,x,0\rangle\langle y,x,0|+|y-1,x,1\rangle\langle y,x,1|\right),\\
\!\!\! T_\rmH & =\sum_x\left(|y,x+1,0\rangle\langle y,x,0|+|y,x-1,1\rangle\langle y,x,1|\right),
\end{aligned}
\end{equation}
which describe the vertical and horizontal walkings of the walker, respectively;  $ C(t)=\sum_{y,x}|y,x\rangle\langle y,x|\otimes C(y,x,t) $ encodes position-dependent coin operators, which control the walking direction in each step. The desired quantum measurement is realized by
engineering the coin operators and measuring the walker position after certain steps.

To realize the genuine three-copy collective measurement in \eref{eq:measurement}, we encode two qubits in the walker's vertical and horizontal positions with $y,x=\pm1$, and encode the third qubit as the coin state. Other positions of the walker are used as an ancilla. Then the genuine three-copy collective measurement is realized by virtue of a nine-step 2D quantum walk with 30 elaborately designed coin operators as explained in  Supplemental Material S1.


\emph{Experimental setup}---Our experimental setup is based on a photonic platform as illustrated in Fig.~\ref{fig:exp} and is composed of two modules designed to prepare thee-copy states and to perform the genuine three-copy collective measurement, respectively. 


The second module implements the genuine three-copy collective measurement in \eref{eq:measurement} based on quantum walks as described above. Here the walker and coin states are represented by the path and polarization degrees of freedom of a photon, respectively. To implement the conditional translation operator $T(t)$, we use beam displacers (BDs) to displace the horizontal polarization (H) component away from the vertical polarization (V) component. 
By adjusting the optical axis of each BD, either horizontal walking or vertical walking can be realized at each step. To realize the 30 coin operators $C(y,x,t)$, we use a number of half-wave plates (HWPs) and quarter-wave plates (QWPs) with educated choices of rotation angles listed in  Supplemental Material S1.  After the quantum walk, the six single-photon-counting modules (SPCMs) M$_1$ to M$_{6}$ correspond to the first six POVM elements $E_1$ to $E_{6}$ in \eref{eq:measurement}, respectively,  while the four SPCMs M$_{7(1)}$ to M$_{7(4)}$ together correspond to the last POVM element $E_7$.

The first module of our setup is designed to prepare three copies of a given qubit state $\rho=|\psi\>\<\psi|$, encoded in three degrees of freedom of a photon, that is, the vertical position, horizontal position, and polarization. A 404-nm laser beam in vertical polarization pumps a 1mm-long beta-barium-borate (BBO) crystal and generates a pair of horizontally polarized photons at 808 nm after a type-1 phase-matched spontaneous parametric down-conversion (SPDC) process. One photon is measured by a SPCM to herald its twin photon as a single-photon source. The three degrees of freedom of the heralded photon are then used to encode the three copies of the target state $|\psi\>$. To prepare the first copy in the vertical path degree of freedom with positions $y=\pm 1$, we first encode it in the polarization degree of freedom using HWPs and QWPs and then transfer it to the vertical path basis using a vertical BD and HWPs. After passing through a HWP with deviation angle $45^{\circ}$ at path $y=-1$ and a HWP with deviation angle $0^{\circ} $ at path $y=1$, the heralded photon is prepared in the state $|\psi\rangle\otimes|x=1\rangle\otimes|H\rangle$. Then the second copy of the target state is prepared in the horizontal path degree of freedom with $x=\pm 1$. Finally, the desired three-copy state $|\psi\rangle^{\otimes 3}$ is prepared by virtue of another pair of HWP and QWP and is then sent into the measurement module described above.

\begin{figure}[t]
	\centering\includegraphics[width=1\linewidth]{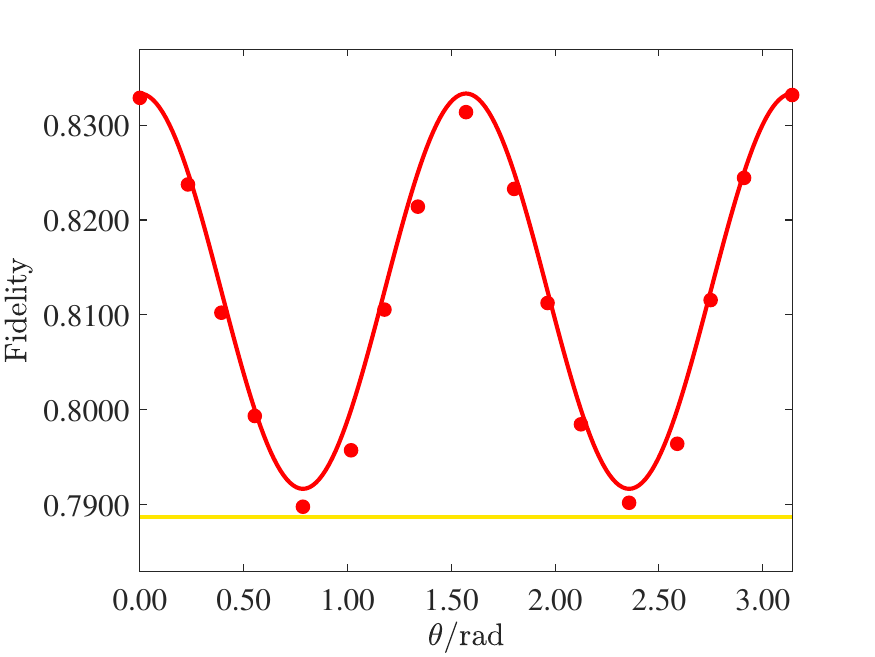}
	\caption{\label{fig:curve}Estimation fidelities achieved by the genuine three-copy collective measurement in \eref{eq:measurement} for input states of the form $|\psi(\theta)\rangle=\cos\theta|0\rangle+\sin\theta|1\rangle$.
		Each data point represents the average over 10 repetitions of 50000 runs each, and the error bar (almost invisible) denotes the standard deviation over the 10 repetitions. As benchmarks, the red line represents the theoretical prediction, and the yellow line represents the 
		upper bound for the average estimation fidelity based on local measurements. }
\end{figure}

\emph{Experimental results}---Before considering optimal quantum state estimation, we implement the genuine three-copy collective measurement in \eref{eq:measurement} and verify its quality by evaluating its fidelity with the ideal measurement. Due to the special structure of the ideal measurement, the fidelity can be evaluated by measuring the six reference states $|\psi_j\rangle^{\otimes 3}$ for $j=1,2,\dots,6$, where $|\psi_j\>$ are defined in \eref{eq:elements} (see  Supplemental Material S2 for details). The six states are prepared and sent into the measurement module 50000 times, and this procedure is repeated 10 times to determine the standard deviation.  The average fidelity determined in the experiment is $0.9942\pm 0.0011$, which 
demonstrates that the genuine collective measurement is realized with very high quality.\par

Next, we consider the estimation of a  random qubit pure state given three identically prepared copies. By virtue of the genuine three-copy collective measurement realized we can implement an optimal estimation protocol. 
If the measurement outcome is $E_j$ for $j=1,2,\ldots,6$, then we choose $|\psi_j\>$ defined in \eref{eq:elements} as an estimator of the true state.  If instead the outcome is $E_7$ (which may occur due to inevitable experimental imperfections), then we choose a random qubit pure state as an estimator. 

\begin{figure}[t]
	\centering\includegraphics[width=1\linewidth]{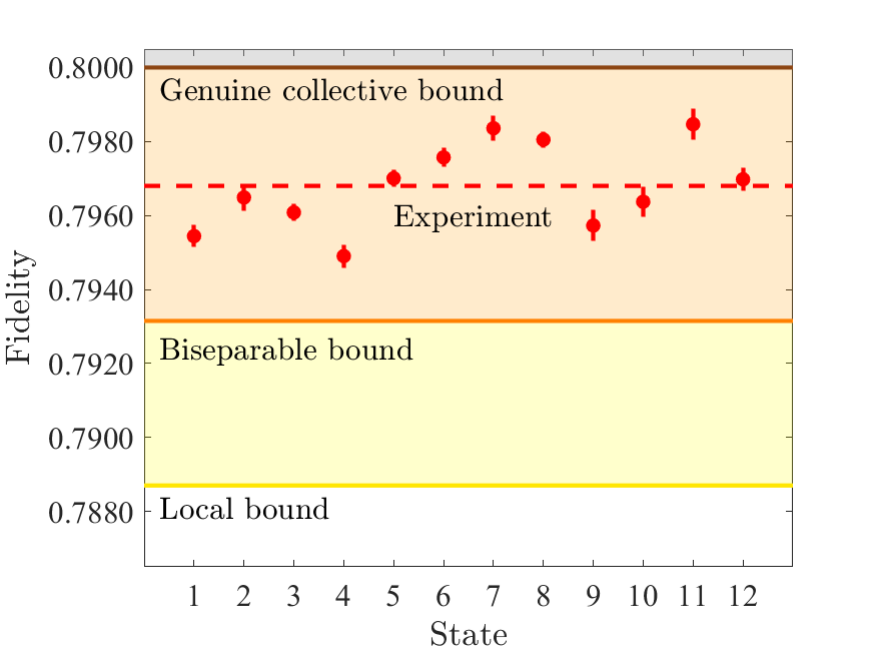}
	\caption{\label{fig:fidelity} 
		Estimation fidelities achieved by the genuine three-copy collective measurement for 12 qubit states that form a regular icosahedron on the Bloch sphere. Each data point represents the average over 10 repetitions of 50000 runs as in Fig.~\ref{fig:curve}, and the error bar denotes the standard deviation. The average estimation fidelity over the 12 states is shown as a red dashed line. 
		Theoretical estimation fidelities for the 12 states all coincide with the genuine collective bound shown as a brown line. The local bound (yellow line)  and biseparable bound (orange line) are also shown as benchmarks. 
	}
\end{figure}

In the first experiment, the input pure state has the form $|\psi(\theta)\rangle=\cos\theta|0\rangle+\sin\theta|1\rangle$ with $\theta=k\pi/16$ for $k=0,1,\ldots, 16$. After performing the genuine collective measurement on three copies of $|\psi(\theta)\rangle$ and obtaining a measurement outcome, an estimator can be constructed and its fidelity can be evaluated. The average fidelity is then determined by repeating this procedure $50000\times 10$ times. The experimental results shown in \fref{fig:curve} agree very well with the theoretical predictions. Notably, the genuine collective measurement can beat the local bound (0.7887) for all states whose Bloch vectors lie on the $x-z$ plane.

In the second experiment, we investigate the average fidelity when the input states are uniformly distributed on the whole Bloch sphere; three identically prepared copies are used in the estimation procedure as before. This uniform distribution can also be replaced by any set of quantum states that forms a 4-design \cite{Zhu22}. For simplicity here we choose 12 states that form a regular icosahedron on the Bloch sphere (see  Supplemental Material S3 for details). 
The estimation fidelity for each state is determined by performing the genuine three-copy collective measurement $50000\times 10$ times as before. Again,
the experimental results agree very well with the theoretical predictions, as shown in Fig.~4. The average estimation fidelity over the 12 states is 0.7968(3), which is much higher than the local bound (0.7887) as expected.  Moreover, the fidelity beats the biseparable bound (0.7932) by more than 11 standard deviations. This result clearly demonstrates that the measurement we realized has a superior information extraction capability over all restricted collective measurements. 
Such a genuine three-copy collective measurement has never been realized before.


\emph{Summary and outlook}---By using a nine-step 2D photonic quantum walk with 30 engineered coin operators, we experimentally implemented a genuine three-copy collective measurement with a fidelity of 0.9942(11). As an application, we realized an optimal estimation protocol given three copies of a random qubit pure state and achieved an average estimation fidelity of 0.7968(3), which beats not only the local bound (0.7887), but also the biseparable bound (0.7932). These results clearly demonstrate the true implementation of a genuine three-copy collective measurement and its power in quantum information processing, which have never been demonstrated before.
Our work represents a major step in understanding genuine multipartite nonclassicality in quantum measurements instead of quantum states and is expected to trigger a cascade of future research works on this intriguing subject. 
More sophisticated  many-copy collective measurements may be realized using programmable photonic chips \cite{arrazola2021quantum,maring2024versatile}, which offer scalability and flexibility. Other platforms based on superconductors, ions, and neutral atoms are also promising for exploring the rich physics of collective measurements.



\section{Acknowledgements}
The work at the University of Science and Technology of China is supported by the Innovation Program for Quantum Science and Technology (Grant No.~2023ZD0301400), the National Natural Science Foundation of China (Grants No.~62222512, No.~12104439, and No.~12134014), and the Anhui Provincial Natural Science Foundation (Grant No.~2208085J03).  The work at Fudan University is supported by the National Key Research and Development Program of China (Grant No.~2022YFA1404204), Shanghai Science and Technology Innovation Action Plan (Grant No.~24LZ1400200), Shanghai Municipal Science and Technology Major Project (Grant No.~2019SHZDZX01), Innovation Program for Quantum Science and Technology (Grant No.~2024ZD0300101), and the National Natural Science Foundation of China (Grant No.~92165109).

\appendix


\hyphenation{ALPGEN}
\hyphenation{EVTGEN}
\hyphenation{PYTHIA}


\onecolumngrid
 \clearpage
 \newpage

 \setcounter{equation}{0}
 \setcounter{figure}{0}
 \setcounter{table}{0}
 \setcounter{section}{0}
 \makeatletter
 \renewcommand{\theequation}{S\arabic{equation}}
 \renewcommand{\thefigure}{S\arabic{figure}}
 \renewcommand{\thetable}{S\arabic{table}}
 \renewcommand{\thesection}{S\arabic{section}}
\onecolumngrid
 \begin{center}
 	\textbf{\large Experimental Realization of Genuine Three-Copy Collective Measurements for Optimal Information Extraction: Supplemental Material}
 \end{center}
 \centerline{(Dated: May 8, 2025)}
 \bigskip
In this Supplemental Material we present more details and some auxiliary results about the optimal POVM $\scrE$ defined in Eq.~(1) in the main text. 
In Secs.~S1 and S2 we explicate the realization and verification of $\scrE$. Section~S3 further provides the Bloch vectors of the 12 states employed for determining the estimation  fidelity of $\scrE$. In Sec.~S4 we introduce formal definitions of biseparable POVMs and genuine collective POVMs. In Sec.~S5 we show that $\scrE$  is genuinely collective although all its POVM elements are biseparable.

\section{S1. Realization of the genuine three-copy collective measurement}
The genuine three-copy collective measurement in Eq.~(1) can be realized by a nine-step 2D photonic quantum walk as shown in Fig.~2 in the main text and
Fig.~\ref{fig:exp} below. The coin operators and conditional translation operators featured in this scheme read
\begin{equation}
    \begin{gathered}                         C(\mathrm{H_1})=\left(\begin{array}{cc}
            1 & 0 \\
            0 & -1
        \end{array}\right),\quad  
        C(\mathrm{H_2})=\left(\begin{array}{cc}
            0 & 1 \\
            1 & 0
        \end{array}\right),\quad    
        C(\mathrm{H_3})=\frac 1{\sqrt 3}\left(\begin{array}{cc}
            \sqrt 2 & 1 \\
            1 & -\sqrt 2
        \end{array}\right),\\[1ex]
        C(\mathrm{H_4})=\frac12\left(\begin{array}{cc}
            -\sqrt 3 & 1 \\
            1 & \sqrt 3
        \end{array}\right),\quad    
        C(\mathrm{H_5})=\frac12\left(\begin{array}{cc}
            -\sqrt 3 & -1 \\
            -1 & \sqrt 3
        \end{array}\right),\quad    
        C(\mathrm{H_6})=\frac1{\sqrt 2}\left(\begin{array}{cc}
            1 & 1 \\
            1 & -1
        \end{array}\right),\\[1ex]
        C(\mathrm{H_7})=\frac1{\sqrt 2}\left(\begin{array}{cc}
            -1 & 1 \\
            1 & 1
        \end{array}\right),\quad    
        C(\mathrm{H_8})=\frac 1{\sqrt 3}\left(\begin{array}{cc}
            -\sqrt 2 & 1 \\
            1 & \sqrt 2
        \end{array}\right),\quad    
        C(\mathrm{H_9})=-\frac i{\sqrt 2}\left(\begin{array}{cc}
            1 & 1 \\
            1 & -1
        \end{array}\right),\\[1ex]
        T(1)=T(2)=T(4)=T(5)=T(8)=T(9)=T_\rmV, \quad
        T(3)=T(6)=T(7)=T_\rmH.\\
    \end{gathered}
\end{equation} 
The four coin operators $C(\mathrm{H_2})$, $C(\mathrm{H_1})$, $C(\mathrm{H_7})$, and  $C(\mathrm{H_3})$  are used  14, 4, 4, and 3 times, respectively, while the other five coin operators are used only once. In total,  30 coin operators are employed. These coin operators are realized by adjusting the deviation angles of HWPs as listed in the table embedded in \fref{fig:exp} and phase differences between individual paths. Two types of conditional translation operators are realized by adjusting the optical-axis directions of the BDs. 

\begin{figure*}[b]	\centering\includegraphics[width=1\linewidth]{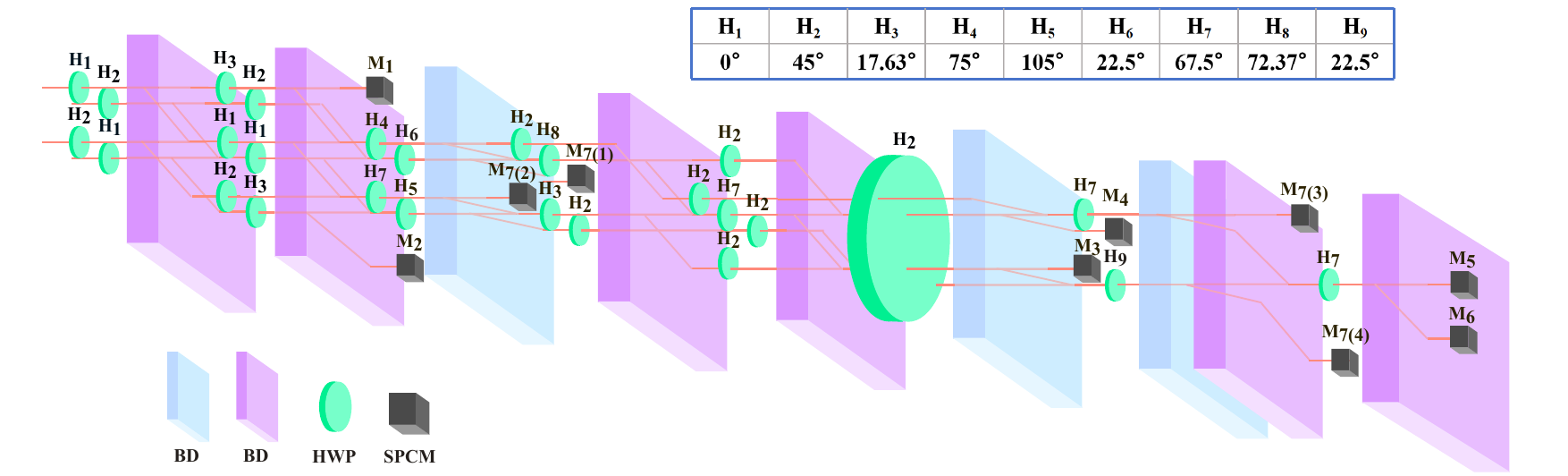}
	\caption{\label{fig:exp}  Schematic diagram of the experimental setup for realizing the genuine three-copy collective measurement in Eq.~(1). Blue BDs and purple BDs are used to realize the conditional translation operators along horizontal and vertical directions, respectively. 
 HWPs with deviation angles specified in the embedded table are used to realize the 30 coin operators. 
  The six SPCMs $\mathrm{M_1}$ to $\mathrm{M_{6}}$ correspond to the first six outcomes in Eq.~(1), while
 the four SPCMs $\mathrm{M_{7(1)}}$ to $\mathrm{M_{7(4)}}$ together correspond to the last outcome in Eq.~(1).}
\end{figure*}

The evolution of the system state in each step is governed by a unitary transformation of the form $U(t) = T(t)C(t)$, where $ C(t)=\sum_{y,x}|y,x\rangle\langle y,x|\otimes C(y,x,t) $. Accordingly, 
the  unitary transformation associated with the first $t$ steps has the form $U_t = T(t)C(t)\cdots T(2)C(2)T(1)C(1)$. Suppose $\rho^{\otimes 3}$ is input as the initial state and we measure the walker position after $t$ steps. Let $P$ be the projector onto the subspace that  is associated  with the three qubits and is spanned by the eight kets,
\begin{align}
  |\pm1, \pm1, 0\>,\quad  |\pm1, \pm1, 1\>.
\end{align}
Then the probability of obtaining outcome $(y=j,x=k)$ is given by 
\begin{equation}
    p=\tr \bigl[U_t \rho^{\otimes 3} U_t^{\dag}(|j,k\rangle\langle j, k|\otimes \bbone)\bigr]=\tr\bigl(\Pi \rho^{\otimes 3}\bigr),
\end{equation}
where $\bbone$ is the identity operator for the coin Hilbert space, and $\Pi=P U_t^{\dag}(|j,k\rangle\langle j, k|\otimes \bbone)U_t P$ is the  POVM element corresponding to the position $(y=j,x=k)$.

To verify that the measurement of the walker position can effectively realize the genuine collective measurement in Eq.~(1), we first determine the evolution operators tied to the first step $(t=1)$ and first two steps $(t=2)$, with the results
\begin{align}
    U_1P
    &=|2,1,0\rangle\langle 1,1,0|+|0,1,0\rangle\langle -1,1,1|-|0,1,1\rangle\langle 1,1,1|+|-2,1,1\rangle\langle -1,1,0|+|2,-1,0\rangle\langle 1,-1,1| \nonumber\\
    &\quad +|0,-1,0\rangle\langle -1,-1,0|+|0,-1,1\rangle\langle 1,-1,0|-|-2,-1,1\rangle\langle -1,-1,1|,\\
U_2P&=\sqrt{\frac23}|3,1,0\rangle\langle 1,1,0|+\sqrt{\frac13}|1,1,1\rangle\langle 1,1,0|+|1,1,0\rangle\langle -1,1,1|+|-1,1,1\rangle\langle 1,1,1|+|-1,1,0\rangle\langle -1,1,0|\nonumber\\
    &\quad 
    +|1,-1,1\rangle\langle 1,-1,1|+|1,-1,0\rangle\langle -1,-1,0|-|-1,-1,1\rangle\langle 1,-1,0|-\sqrt{\frac13}|-1,-1,0\rangle\langle -1,-1,1|\nonumber\\
    &\quad +\sqrt{\frac23}|-3,-1,1\rangle\langle -1,-1,1|.
\end{align}
Accordingly, the two POVM elements associated with  $\mathrm{M_1}$ and $\mathrm{M_2}$ at positions $(3,1)$ and $(-3,-1)$  read
\begin{align}
    \Pi_1&=PU_2^{\dag}(|3,1\rangle\langle 3, 1|\otimes \bbone)U_2P=\frac23|1,1,0\rangle\langle 1,1,0|=\frac23(|\psi_1\rangle\langle\psi_1|)^{\otimes 3},\\
    \Pi_2&=PU_2^{\dag}(|-3,-1\rangle\langle -3, -1|\otimes \bbone)U_2P=\frac23|-1,-1,1\rangle\langle -1,-1,1|=\frac23(|\psi_2\rangle\langle\psi_2|)^{\otimes 3}.
\end{align}
From now on, there is no need to consider the components at positions $(3,1)$ and $(-3,-1)$. 


The evolution operator tied to the first three steps reads
\begin{align}   U_3P&=\sqrt{\frac1{12}}|1,2,0\rangle\left(\langle 1,1,0|-3\langle -1,1,1|\right)+\frac12|1,0,1\rangle\left(\langle 1,1,0|+\langle -1,1,1|\right)\nonumber\\
    &\quad +\sqrt{\frac12}|-1,2,0\rangle\left(\langle 1,1,1|-\langle -1,1,0|\right)+\sqrt{\frac12}|-1,0,1\rangle(\langle 1,1,1|+\langle-1,1,0|)\nonumber\\
    &\quad +\sqrt{\frac{1}{2}}|1,0,0\rangle(\langle 1,-1,1|+\langle-1,-1,0|)+\sqrt{\frac12}|1,-2,1\rangle(\langle -1,-1,0|-\langle 1,-1,1|)\nonumber\\
    &\quad +\frac12|-1,0,0\rangle(\langle 1,-1,0|+\langle -1,-1,1)+\sqrt{\frac1{12}}|-1,-2,1\rangle(\langle -1,-1,1|-3\langle1,-1,0|).
\end{align}
Accordingly, the POVM elements associated with $\mathrm{M_{7(1)}}$ and $\mathrm{M_{7(2)}}$  read
\begin{align}
    \Pi_7&=PU_3^{\dag}(|1,-2\rangle\langle 1,-2|\otimes \bbone)U_3P=\frac12(|-1,-1,0\rangle-|1,-1,1\rangle)(\langle -1,-1,0|-\langle 1,-1,1|),\\
    \Pi_8&=PU_3^{\dag}(|-1,2\rangle\langle -1,2|\otimes \bbone)U_3P=\frac12(|1,1,1\rangle-|-1,1,0\rangle)(\langle 1,1,1|-\langle -1,1,0|).
\end{align}

Similarly, after step $t=6$ we have
\begin{align}    U_6P&=\sqrt{\frac1{12}}|1,1,1\rangle\left(\langle 1,1,0|-3\langle -1,1,1|\right)+\sqrt{\frac1{12}}|-1,-1,0\rangle(\langle -1,-1,1|-3\langle1,-1,0|)\nonumber\\
    &+\sqrt{\frac1{12}}|1,1,0\rangle\left(\langle 1,1,0|-2\langle 1,-1,1|+\langle -1,1,1|-2\langle -1,-1,0|\right)\nonumber\\
    &+\sqrt{\frac1{12}}|-1,-1,1\rangle(\langle -1,-1,1|-2\langle -1,1,0|+\langle 1,-1,0|-2\langle 1,1,1|)\nonumber\\
    &+\sqrt{\frac1{12}}|-1,1,0\rangle(\langle 1,1,1|+\langle 1,1,0|+\langle 1,-1,1|+\langle 1,-1,0|+\langle -1,1,1|+\langle -1,1,0|+\langle -1,-1,1|+\langle -1,-1,0|)\nonumber\\
    &+\sqrt{\frac1{12}}|1,-1,1\rangle(-\langle 1,1,1|+\langle 1,1,0|+\langle 1,-1,1|-\langle 1,-1,0|+\langle -1,1,1|-\langle -1,1,0|-\langle -1,-1,1|+\langle -1,-1,0|);
\end{align}
the POVM elements associated with $\mathrm{M_{3}}$ and $\mathrm{M_{4}}$ read
\begin{align}
    \Pi_3&=PU_6^{\dag}(|-1,1\rangle\langle -1, 1|\otimes \bbone)U_6P=\frac23(|\psi_3\rangle\langle\psi_3|)^{\otimes 3},\\
    \Pi_4&=PU_6^{\dag}(|1,-1\rangle\langle 1, -1|\otimes \bbone)U_6P=\frac23(|\psi_4\rangle\langle\psi_4|)^{\otimes 3}.
\end{align}

After step $t=8$ we have
\begin{align} U_8P&=\sqrt{\frac16}|2,2,0\rangle(\langle 1,-1,1|-2\langle-1,1,1|+\langle-1,-1,0|)-i\sqrt{\frac16}|-2,-2,1\rangle(\langle -1,1,0|-2\langle1,-1,0|+\langle1,1,1|)\nonumber\\
    &\quad -\sqrt{\frac16}|0,0,1\rangle(-\langle 1,1,0|+\langle1,-1,1|+\langle-1,1,1|+\langle-1,-1,0|)\nonumber\\
    &\quad +i\sqrt{\frac16}|0,0,0\rangle(-\langle -1,-1,1|+\langle-1,1,0|+\langle1,-1,0|+\langle1,1,1|);
\end{align}
the POVM elements associated  
with $\mathrm{M_{7(3)}}$ and $\mathrm{M_{7(4)}}$ read
\begin{align}
 \!   \Pi_9&=PU_8^{\dag}(|2,2\rangle\langle 2, 2|\otimes \bbone)U_8P
    =\frac16(| 1,-1,1\rangle-2|-1,1,1\rangle+|-1,-1,0\rangle)(\langle 1,-1,1|-2\langle-1,1,1|+\langle-1,-1,0|),\\
  \!  \Pi_{10}&=PU_8^{\dag}(|-2,-2\rangle\langle -2, -2|\otimes \bbone)U_8P
    =\frac16(| -1,1,0\rangle-2|1,-1,0\rangle+|1,1,1\rangle)(\langle -1,1,0|-2\langle1,-1,0|+\langle1,1,1|).
\end{align}

After step $t=9$ we have
\begin{equation}
\begin{aligned}
    U_9P&=\sqrt{\frac1{12}}|1,0,0\rangle(i\langle 1,1,1|+\langle 1,1,0|-\langle 1,-1,1|+i\langle 1,-1,0|-\langle -1,1,1|+i\langle -1,1,0|-i\langle -1,-1,1|-\langle -1,-1,0|)\\
    &-\sqrt{\frac1{12}}|-1,0,1\rangle(i\langle 1,1,1|-\langle 1,1,0|+\langle 1,-1,1|+i\langle 1,-1,0|+\langle -1,1,1|+i\langle -1,1,0|-i\langle -1,-1,1|+\langle -1,-1,0|);
\end{aligned}
\end{equation}
the POVM elements associated
with $\mathrm{M_{5}}$ and $\mathrm{M_{6}}$ read
\begin{align}
    \Pi_5&=PU_9^{\dag}(|1,0\rangle\langle 1, 0|\otimes \bbone)U_9P=\frac23(|\psi_5\rangle\langle\psi_5|)^{\otimes 3},\\
    \Pi_6&=PU_9^{\dag}(|-1,0\rangle\langle -1, 0|\otimes \bbone)U_9P=\frac23(|\psi_6\rangle\langle\psi_6|)^{\otimes 3}.
\end{align}

Notably, $\Pi_j=E_j$ for $j=1,2,\dots,6$ and $\sum_{j=7}^{10}\Pi_j=E_7$. Therefore, the nine-step 2D quantum walk illustrated in Fig.~\ref{fig:exp} can indeed realize the collective measurement in Eq.~(1).

\section{S2. Verification of the implemented genuine three-copy collective measurement}
Suppose $\{A_j\}_j$ and $\{A'_j\}_j$ are two POVMs on a given $d$-dimensional Hilbert space and have the same number of POVM elements. Following \rcite{hou2018deterministic} the fidelity between  $\{A_j\}_j$ and $\{A'_j\}_j$  is defined as
\begin{equation}
F=\left(\tr\sqrt{\sqrt{\sigma}\sigma'\sqrt{\sigma}}\right)^2,
\end{equation}
where 
\begin{align}
\sigma=\frac1d\sum_jA_j\otimes|j\rangle\langle j|,\quad  \sigma'=\frac1d\sum_jA'_j\otimes|j\rangle\langle j|,
\end{align}
and $\{|j\rangle\}_j$ forms an orthonormal basis for an ancillary system.

Here we are interested in the fidelity between the POVM $\scrE=\{E_j\}_{j=1}^7$ defined in Eq.~(1) and its experimental realization $\scrE'=\{E'_j\}_{j=1}^7$. Note that any three-copy state involved in the estimation problem studied in this work is supported in the tripartite symmetric subspace $\Sym_3(\caH)$, where $\caH$ is the Hilbert space for a single qubit. So we can focus on the realization of the effective POVM on this subspace, that is, $\{E_j\}_{j=1}^6$. Thanks to the special structure of the target POVM $\{E_j\}_{j=1}^6$, the fidelity of  $\{E'_j\}_{j=1}^6$  can be estimated by measuring six reference states $|\psi_j\rangle^{\otimes 3}$ for $j=1,2,\dots,6$, where $|\psi_j\rangle$ are defined in Eq.~(2). By definition it is straightforward to verify that
\begin{align}
\sum_{j=1}^6 (|\psi_j\rangle\langle\psi_j|)^{\otimes 3}=\frac32\sum_{j=1}^6 E_j=
\frac32 P_3, 
\end{align}
where $P_3$ is the projector onto $\Sym_3(\caH)$. Let $p_{jk}=\tr[(|\psi_j\rangle\langle\psi_j|)^{\otimes 3}E_k]$ be the probability of obtaining outcome $k$ given the reference state $|\psi_j\rangle^{\otimes 3}$ and the ideal measurement, and let $p'_{jk}=\tr[(|\psi_j\rangle\langle\psi_j|)^{\otimes 3}E'_k]$ be the counterpart associated with the actual POVM realized. Then we have  
\begin{align}
\sum_j p_{jk}=\frac32\tr E_k=1,\quad \sum_j p'_{jk}=\frac32\tr (P_3E'_k). 
\end{align}
Moreover,
\begin{align}   F&=\left(\frac1d\sum_{j=1}^6\tr \sqrt{\sqrt{E_j}E'_j\sqrt{E_j}}\right)^2=\left(\frac14\sum_{j=1}^6\sqrt{\tr (E_j E'_j)}\right)^2=\left(\frac14\sum_{j=1}^6\sqrt{\frac23\tr \left[(|\psi_j\rangle\langle\psi_j|)^{\otimes 3} E'_j\right]}\right)^2
\nonumber\\
&=\left(\frac14\sum_{j=1}^6\sqrt{\frac23p_{jj}'}\right)^2=\frac{1}{24} \left(\sum_{j=1}^6\sqrt{p_{jj}'}\right)^2,
\end{align}
where $d=4$ is the dimension of $\Sym_3(\caH)$. Here the second and third equalities hold because $E_j=\frac{2}{3}(|\psi_j\>\<\psi_j|)^{\otimes3}$ is rank 1 for $j=1,2,\ldots,6$. This equation means the fidelity of $\scrE'=\{E'_j\}_{j=1}^6$ is determined by the probabilities $p_{jj}'$ for $j=1,2,\ldots, 6$, which can be approximated by frequencies in experiments.

\section{S3. 12 states used to determine the estimation fidelity of the genuine collective measurement}
To determine the average estimation fidelity of the genuine three-copy collective measurement characterized by the POVM $\scrE$ defined in Eq.~(1) in the main text, we employed 12 qubit states whose Bloch vectors form the vertices of a regular icosahedron. To be concrete, the coordinates of these  vertices (Bloch vectors) read
\begin{equation}
\begin{aligned}
    \vec n_1&=\frac1{\sqrt{1+g^2}}(1 ,g, 0),&\quad \vec n_2&=\frac1{\sqrt{1+g^2}}(1 ,-g, 0),\\ \vec n_3&=\frac1{\sqrt{1+g^2}}(-1 ,g, 0),&\quad \vec n_4&=\frac1{\sqrt{1+g^2}}(-1 ,-g, 0),\\
    \vec n_5&=\frac1{\sqrt{1+g^2}}(0, 1 ,g),&\quad \vec n_6&=\frac1{\sqrt{1+g^2}}(0, 1 ,-g),\\ \vec n_7&=\frac1{\sqrt{1+g^2}}(0, -1 ,g),&\quad \vec n_8&=\frac1{\sqrt{1+g^2}}(0, -1 ,-g),\\
    \vec n_9&=\frac1{\sqrt{1+g^2}}(g, 0, 1),&\quad \vec n_{10}&=\frac1{\sqrt{1+g^2}}(-g, 0, 1),\\ \vec n_{11}&=\frac1{\sqrt{1+g^2}}(g, 0, -1),&\quad \vec n_{12}&=\frac1{\sqrt{1+g^2}}(-g, 0, -1),\\
\end{aligned}
\end{equation}
where $g=(1+\sqrt{5})/2$.


\section{S4. Biseparable measurements and genuine collective measurements}
To appreciate the significance of the optimal collective measurement characterized by the POVM $\scrE$ defined in Eq.~(1), here we provide formal definitions of biseparable POVMs (measurements) and genuine collective POVMs (measurements); see the companion paper \cite{YiZHX24} for more details.

Consider an $N$-partite quantum system with  total Hilbert space 
\begin{equation}\label{eq:HTtensorDecom}
    \caH_\rmT = \caH_1\otimes \caH_2\otimes \cdots\otimes \caH_N, 
\end{equation}
where $\caH_r$ is the Hilbert space of party $r$ for $r\in [N]:=\{1,2,\ldots, N\}$. A partition $\caP=\{I_1, I_2, \ldots, I_m\}$ of $[N]$ is a set of at least two disjoint nonempty subsets of $[N]$ whose union is $[N]$; it can also be written as $(I_1|I_2|\cdots |I_m)$. Here the order of $I_j$ can be changed without modifying the partition. 
The partition $\caP$ is a bipartition if $|\caP|=m=2$. With respect to the partition $\caP$ the Hilbert space $\caH_\rmT$ can be decomposed as follows:
\begin{equation}
	\caH_\rmT = \bigotimes_{k=1}^{m}\caH_{I_k},\quad \caH_{I_k} : = \bigotimes_{r\in I_k}\caH_r. 
\end{equation}
Here we assume that the tensor factors in the expansion of  $\bigotimes_{k=1}^{m}\caH_{I_k}$ follow the order in \eref{eq:HTtensorDecom}. A positive operator  $A$ on $\caH_\rmT$ is $\caP$ separable if it can be expressed as follows:
\begin{align}
A=\sum_l A_l^{I_1}\otimes A_l^{I_2}\otimes \cdots\otimes A_l^{I_m},
\end{align}
where $A_l^{I_k}$ is a positive operator on $\caH_{I_k}$ for each $l$ and  $k=1,2,\ldots,m$. 
 The  operator  $A$ is \emph{biseparable} if it can be expressed as 
\begin{align}
A=\sum_{\caP,\,|\caP|=2}A_\caP,
\end{align}
where the summation runs over all bipartitions of $[N]$ and $A_\caP$ is $\caP$ separable. 


A POVM $\scrA=\{A_j\}_j$ (and the corresponding measurement) on $\caH_\rmT$ 
is $\caP$ separable if every POVM element $A_j$ is $\caP$ separable. The POVM $\scrA$   is \emph{biseparable} if it is a coarse-graining  of a POVM \cite{MartM90,Zhu22} of the form 
\begin{align}
	\scrK=\bigsqcup_{\caP,\, |\caP|=2} p_\caP\scrK_\caP,
\end{align}
where $\scrK_\caP$ is a POVM on $\caH_\rmT$ that is $\caP$ separable,   $p_\caP\scrK_\caP$ means an element-wise product, $\{p_\caP\}_\caP$ forms a probability distribution, and the notation $\bigsqcup_{\caP,\, |\caP|=2}$ means the disjoint union over all bipartitions of $[N]$. By construction $\scrK$ is a convex combination of $\scrK_\caP$ associated with bipartitions and
can be realized by performing each $\scrK_\caP$ with probability $p_\caP$, then the POVM $\scrA$ can be realized via data processing. 
By definition all POVM elements of a biseparable POVM are biseparable, but the converse is not guaranteed automatically. A POVM is  \emph{genuinely collective} if it is not biseparable.

Next, we provide several concrete examples of biseparable POVMs. Let $|\Phi^+\rangle, |\Phi^-\rangle, |\Psi^+\rangle, |\Psi^-\rangle$ be four Bell states that form a Bell basis for a two-qubit system:
\begin{align}
    |\Phi^\pm\>=\frac{|00\>\pm |11\>}{\sqrt{2}},\quad |\Psi^\pm\>=\frac{|01\>\pm|10\>}{\sqrt{2}}. 
\end{align}
Let $\scrA= \{A_j\}^8_{j=1}$ and $\scrB = \{B_j\}^8_{j=1}$ be two three-qubit POVMs defined as follows:
\begin{equation}
    \begin{aligned}                  
     A_1=|\Phi^+\rangle\langle\Phi^+|\otimes|0\rangle\langle0|,\quad A_2=|\Phi^-\rangle\langle\Phi^-|\otimes|0\rangle\langle0|, \quad        A_3=|\Psi^+\rangle\langle\Psi^+|\otimes|0\rangle\langle0|,\quad A_4=|\Psi^-\rangle\langle\Psi^-|\otimes|0\rangle\langle0|,\\
        A_5=|\Phi^+\rangle\langle\Phi^+|\otimes|1\rangle\langle1|,\quad A_6=|\Phi^-\rangle\langle\Phi^-|\otimes|1\rangle\langle1|,\quad        A_7=|\Psi^+\rangle\langle\Psi^+|\otimes|1\rangle\langle1|,\quad A_8=|\Psi^-\rangle\langle\Psi^-|\otimes|1\rangle\langle1|,  \\            
        B_1=|0\rangle\langle0|\otimes|\Phi^+\rangle\langle\Phi^+|,\quad B_2=|0\rangle\langle0|\otimes|\Phi^-\rangle\langle\Phi^-|,\quad        B_3=|0\rangle\langle0|\otimes|\Psi^+\rangle\langle\Psi^+|,\quad B_4=|0\rangle\langle0|\otimes|\Psi^-\rangle\langle\Psi^-|,\\
        B_5=|1\rangle\langle1|\otimes|\Phi^+\rangle\langle\Phi^+|,\quad B_6=|1\rangle\langle1|\otimes|\Phi^-\rangle\langle\Phi^-|,\quad         B_7=|1\rangle\langle1|\otimes|\Psi^+\rangle\langle\Psi^+|,\quad B_8=|1\rangle\langle1|\otimes|\Psi^-\rangle\langle\Psi^-|.
    \end{aligned}
\end{equation}
Note that  $\scrA$ is $(12|3)$ separable, while $\scrB$ is $(23|1)$ separable (here $1$ and  $12$  are shorthands for $\{1\}$ and $\{1,2\}$, respectively, and likewise for $2$, $3$, $13$, and $23$). By virtue of $\scrA$ and $\scrB$ we can construct more complicated biseparable POVMs. Here are two specific examples:
\begin{align}
 \scrK_1=\{p A_j\}^8_{j=1}\cup \{(1-p)B_j\}^8_{j=1},\quad 
 \scrK_2=\{p A_j+(1-p)B_j\}^8_{j=1},\quad 0<p<1,
\end{align}
where  $\scrK_1$ is a convex combination of  $\scrA$ and $\scrB$, while  $\scrK_2$ is a coarse-graining of $\scrK_1$.  The POVM $\scrK_1$ can be realized by performing 
$\scrA$ and $\scrB$ with probabilities $p$ and  $1-p$, respectively; then the POVM $\scrK_2$ can be realized if we ignore which POVM is performed. By definition, both  $\scrK_1$ and  $\scrK_2$ are biseparable. In contrast, the POVM $\scrE$ presented in Eq.~(1) in the main text is genuinely collective as shown in the next section.

\section{S5. A genuine collective POVM 
composed of biseparable POVM elements }
Here we show that the POVM $\scrE$ defined in Eq.~(1) in the main text is genuinely collective although all its POVM elements are biseparable.

Note that the POVM element $E_7$ in Eq.~(1) can be regarded as an operator acting on $\caH^{\otimes3}$, where $\caH$ is the Hilbert space for a single qubit.  Let $\bbW$ be the unitary operator on $\caH^{\otimes3}$ that is associated with a cyclic permutation of the three subsystems. Let  $P_2^{\caA}$ be the projector onto
the antisymmetric subspace in $\caH^{\otimes 2}$ and  $\Pi = P_2^{\caA} \otimes \bbone$.
Then the POVM element $E_7$  can be expressed as follows:
\begin{equation}
E_7=\frac23(\Pi+\bbW^{\dagger}\Pi \bbW+\bbW\Pi \bbW^{\dagger}),
\end{equation}
so  $E_7$ is biseparable and invariant under the action of $\bbW$. It follows that all POVM elements in $\scrE$ are biseparable. 



Next, we prove that the POVM $\scrE$ defined in Eq.~(1) is genuinely collective. Suppose, by way of contradiction, that $\scrE$ is biseparable. Then $\scrE$ can be expressed as a coarse-graining of a POVM of the form 
\begin{align}
\scrK=p_3 \scrK_{(12|3)}\sqcup p_2\scrK_{(13|2)}\sqcup p_1 \scrK_{(23|1)},
\end{align}
where $p_1, p_2, p_3$ form a probability distribution, and the three POVMs $\scrK_{(12|3)}$, $\scrK_{(13|2)}$, and $\scrK_{(23|1)}$ are biseparable with respect to the three bipartitions $(12|3)$, $(13|2)$, and $(23|1)$, respectively. In addition, we can assume that all POVM elements in the three POVMs are rank-1 and that $p_3>0$ without loss of generality.


Since every POVM element in  $\scrE$ is  supported either in the symmetric subspace $\Sym_3(\caH)$ in $\caH^{\otimes 3}$ or in its orthogonal complement $\Sym_3(\caH)^\perp$, every POVM element in  $\scrK_{(12|3)}$ has the same property. Let $\scrK_{(12|3)}'$ be the subset of POVM elements in  $\scrK_{(12|3)}$ that are supported in  $\Sym_3(\caH)^\perp$; then $\scrK_{(12|3)}'$ can be regarded as  a POVM on $\Sym_3(\caH)^\perp$, which implies that
\begin{align}\label{eq:POVMgcolProof}
    \sum_{K\in \scrK_{(12|3)}'}K=\bbone-P_3, 
\end{align}
where $P_3$ is the projector onto $\Sym_3(\caH)$. Therefore, $\tr(P_3K)=0$ for all $K\in \scrK_{(12|3)}'$. Meanwhile, by assumption each POVM element in $\scrK_{(12|3)}'$ has the form $w_j|\Phi_j\>\<\Phi_j|\otimes |\varphi_j\>\<\varphi_j|$, where $w_j> 0$, $|\Phi_j\>\in \caH^{\otimes 2}$, and $|\varphi_j\>\in\caH$. According to \lref{lem:P3orthogonal} below, we have $|\Phi_j\>\<\Phi_j|=P_2^\caA$, so every POVM element in $\scrK_{(12|3)}'$ is supported in $\Sym_2(\caH)^\perp\otimes \caH$, which has dimension 2. Consequently, $\sum_{K\in \scrK_{(12|3)}'}K$ has rank at most 2, which contradicts \eref{eq:POVMgcolProof}. This contradiction shows  that  $\scrE$ cannot be biseparable and is thus genuinely collective.

In the rest of this section, we prove an auxiliary lemma employed in the above proof. 

\begin{lemma}\label{lem:P3orthogonal}
	Suppose $\caH$ has dimension 2, $|\Phi\>\in \caH^{\otimes 2}$, and $|\varphi\>\in\caH$. Then $\tr[P_3(|\Phi\>\<\Phi|\otimes |\varphi\>\<\varphi|)]=0$ iff $|\Phi\>\<\Phi|=P_2^\caA$. 
\end{lemma}
\begin{proof}
	Let $\bbW_{(12)}$ be the swap operator acting on the first two parties of $\caH^{\otimes 3}$. If $|\Phi\>\<\Phi|=P_2^\caA$, then 
	\begin{equation}
		\tr[P_3(|\Phi\>\<\Phi|\otimes |\varphi\>\<\varphi|)]=\tr\bigl[P_3\bbW_{(12)}(|\Phi\>\<\Phi|\otimes |\varphi\>\<\varphi|)\bigr]=-\tr[P_3(|\Phi\>\<\Phi|\otimes |\varphi\>\<\varphi|)],
	\end{equation}
	which implies that $\tr[P_3(|\Phi\>\<\Phi|\otimes |\varphi\>\<\varphi|)]=0$. 
	
To prove the converse, suppose $\tr[P_3(|\Phi\>\<\Phi|\otimes |\varphi\>\<\varphi|)]=0$.   
	Choose a unitary operator $U$ on $\caH$ such that $U|\varphi\>=|0\>$ and let $|\Phi'\>=U^{\otimes 2}|\Phi\>$. Then we have
	\begin{align}
		0&=\tr[P_3(|\Phi\>\<\Phi|\otimes |\varphi\>\<\varphi|)]=\tr\bigl[P_3U^{\otimes 3}(|\Phi\>\<\Phi|\otimes |\varphi\>\<\varphi|)U^{\dag\otimes 3}\bigr]=\tr[P_3(|\Phi'\>\<\Phi'|\otimes |0\>\<0|)]=\tr(|\Phi'\>\<\Phi'| R),
	\end{align}
	where 
	\begin{align}
		R=\tr_3[P_3(\bbone\otimes \bbone\otimes |0\>\<0|)]=|00\>\<00|+\frac{(|01\>+|10\>)(\<01|+\<10|)}{3}
	+\frac{|11\>\<11|}{3}
	\end{align}
 is a positive operator on $\caH^{\otimes 2}$ and has full support in $\Sym_2(\caH)$. So  $|\Phi'\>$ is necessarily supported in the antisymmetric subspace $\Sym_2(\caH)^\perp$, which is spanned by the singlet $(|01\>-|10\>)/\sqrt{2}$. It follows that $|\Phi\>\<\Phi|=|\Phi'\>\<\Phi'|=P_2^\caA$, which completes the proof of \lref{lem:P3orthogonal}. 
\end{proof}

\end{document}